\documentclass[draftcls,onecolumn,11pt]{IEEEtran}

\usepackage{cite} \usepackage{url} \usepackage{hyperref} \usepackage{graphicx} \usepackage[font=footnotesize]{caption} \usepackage[font=footnotesize]{subcaption}

\usepackage{amsmath,amssymb,amsthm} \usepackage{bm}

\DeclareMathOperator{\diag}{diag} \DeclareMathOperator{\GL}{GL} \DeclareMathOperator{\rank}{rank}  \DeclareMathOperator{\row}{row} \DeclareMathOperator{\col}{col} \DeclareMathOperator{\nul}{nul}  \DeclareMathOperator{\shape}{shape} \DeclareMathOperator{\EV}{E}

\newcommand{\Rate}{\mathrm{R}} \newcommand{\Perr}{\mathrm{P}_{\mathrm{e}}}  

   \newcommand{\dR}{\mathrm{d}_\mathrm{R}}

    \newcommand{\CMMC}{\mathrm{CMMC}}

\newcommand{\tr}{\mathrm{T}}

 \newcommand{\mat}[1]{\begin{bmatrix} #1 \end{bmatrix}}

\newtheorem{theorem}{Theorem} \newtheorem{lemma}[theorem]{Lemma}    \newtheorem{proposition}[theorem]{Proposition}   \newenvironment{example}{\textit{Example: }}{} 

\newcommand{\CC}{\mathbb{C}} \newcommand{\ZZ}{\mathbb{Z}}   \newcommand{\FF}{\mathbb{F}}   \newcommand{\calC}{\mathcal{C}}                 \newcommand{\calT}{\mathcal{T}}    \newcommand{\calX}{\mathcal{X}} \newcommand{\calY}{\mathcal{Y}}                                                       

\begin{document}

\title{On Multiplicative Matrix Channels over Finite Chain Rings}

\author{ Roberto W. N\'obrega, \IEEEmembership{Member, IEEE}, Chen Feng, \IEEEmembership{Member, IEEE}, \\ Danilo Silva, \IEEEmembership{Member, IEEE}, and Bartolomeu F. Uch\^oa-Filho, \IEEEmembership{Senior Member, IEEE}

\thanks{This paper was presented in part at the IEEE International Symposium on Network Coding, Calgary, Alberta, Canada, June 2013. This work was partly supported by CNPq-Brazil.}

\thanks{R.\ W.\ N\'obrega, D.\ Silva, and Bartolomeu F.\ Uch\^oa-Filho are with the Department of Electrical Engineering of the Federal University of Santa Catarina, Brazil. C.\ Feng is with the Department of Electrical and Computer Engineering, University of Toronto, Toronto, Canada. (email: rwnobrega@eel.ufsc.br; cfeng@eecg.toronto.edu; danilo@eel.ufsc.br; uchoa@eel.ufsc.br).}

}

\maketitle

\begin{abstract} Motivated by physical-layer network coding, this paper considers communication in multiplicative matrix channels over finite chain rings. Such channels are defined by the law $Y =A X$, where $X$ and $Y$ are the input and output matrices, respectively, and $A$ is called the transfer matrix. It is assumed a coherent scenario in which the instances of the transfer matrix are unknown to the transmitter, but available to the receiver. It is also assumed that $A$ and $X$ are independent. Besides that, no restrictions on the statistics of $A$ are imposed. As contributions, a closed-form expression for the channel capacity is obtained, and a coding scheme for the channel is proposed. It is then shown that the scheme can achieve the capacity with polynomial time complexity and can provide correcting guarantees under a worst-case channel model. The results in the paper extend the corresponding ones for finite fields. \end{abstract}


\section{Introduction}

A \emph{multiplicative matrix channel} (MMC) over a finite field~$\FF_q$ is a communication channel in which the input $\bm{X} \in \FF_q^{n \times \ell}$ and the output $\bm{Y} \in \FF_q^{m \times \ell}$ are related by \begin{equation} \label{eq:channel-law} \bm{Y} = \bm{A} \bm{X} \end{equation} where $\bm{A} \in \FF_q^{m \times n}$ is called the \emph{transfer matrix}\footnote{Throughout this paper, bold symbols are used to represent random entities, while regular symbols are used for their samples.}. Such channels turn out to be suitable models for the end-to-end communication channel between a source node and a sink node in an error-free, erasure-prone network performing random linear network coding~\cite{Koetter.Medard.03,Ho.06,Koetter.Kschischang.08}. In this context, $\bm{X}$ is the matrix whose rows are the $n$ packets (of length~$\ell$) transmitted by the source node, $\bm{Y}$ is the matrix whose rows are the $m$ packets received by the sink node, and $\bm{A}$ is a matrix whose entries are determined by factors such as the network topology and the random choices of the network coding coefficients. Note that each packet can be viewed as an element of the packet space $W = \FF_q^\ell$, a finite vector space.

The present work considers MMCs over \emph{finite chain rings} (of which finite fields are a special case). The motivation comes from \emph{physical-layer network coding}~\cite{Liew.Zhang.Lu.13}. Indeed, recent results show that the modulation employed at the physical layer induces a ``matched choice'' for the ring to be used in the linear network coding layer~\cite{Feng.13}. For instance, if uncoded quaternary phase-shift keying (QPSK) is employed, then the underlying ring should be chosen as $R = \ZZ_2 [i] = \{ 0, 1, i, 1 + i \}$, which is not a finite field, but a finite chain ring. More generally, this is also true for wireless networks employing \emph{compute-and-forward}~\cite{Nazer.Gastpar.11} over arbitrary nested lattices. In this case, the underlying ring happens to be a \emph{principal ideal domain}~$T$ (typically the integers, $\ZZ$, the Gaussian integers, $\ZZ[i]$, or the Eisenstein integers, $\ZZ[\omega])$, with the corresponding message space~$W$ being a \emph{finite $T$-module}~\cite{Feng.13}. As such, \[ W \cong T / \langle d_1 \rangle \times T / \langle d_2 \rangle \times \cdots \times T / \langle d_\ell \rangle, \] where $d_1, d_2, \ldots, d_\ell \in T$ are non-zero non-unit elements satisfying $d_1 \mid d_2 \mid \cdots \mid d_\ell$. A special situation commonly found in practice is when the $d_i$s are all powers of a given prime of~$T$. In this case, the underlying ring can be taken as the finite chain ring $R = T / \langle d_\ell \rangle$, while the message space~$W$ can be seen as a finite $R$-module.

Finite-field MMCs have been studied under an information-theoretic approach according to different assumptions on the probability distribution of the transfer matrix~\cite{Jafari.11,Silva.10,Yang.10.arXiv,Yang.10.ISIT,Nobrega.13}. In this work, following parts of~\cite{Yang.10.arXiv,Yang.10.ISIT}, we consider finite-chain-ring MMCs under a \emph{coherent scenario}, meaning that we assume that the instances of the transfer matrix~$\bm{A}$ are unknown to the transmitter (but available to the receiver). Besides that, we impose no restrictions on the statistics of~$\bm{A}$, except that $\bm{A}$ must be independent of~$\bm{X}$. Furthermore, we are also interested in codes that guarantee reliable communication with a single use of the channel, in the same fashion as~\cite{Silva.08,Silva.Kschischang.11}.

As contributions, we obtain a closed-form expression for the channel capacity, and we propose a coding scheme that combines several codes over a finite field to obtain a code over a finite chain ring. We then show that the scheme can achieve the channel capacity with polynomial time complexity, and that it does not necessarily require the complete knowledge of the probability distribution of $\bm{A}$ [only the expected value of its rank (or, rather, its ``shape''---see Section~\ref{sec:chain-rings}) is needed]. We also present a necessary and sufficient condition under which a code can correct shape deficiencies of the transfer matrix, and we show that the proposed coding scheme can also yield codes with suitable shape-deficiency correction guarantees. Finally, we adapt the coding scheme to the non-coherent scenario, in which the instances of the transfer matrices are unknown to both the transmitter and receiver. Our results extend (and make use of) some of those obtained by Yang et al.\ in~\cite{Yang.10.arXiv,Yang.10.ISIT} and Silva et al.\ in~\cite{Silva.08,Silva.Kschischang.11}, which address the finite field case. It is also worth mentioning that a generalization of the results in~\cite{Silva.10} from finite fields to finite chain rings is presented in~\cite{Feng.13.arXiv}.

The remainder of this paper is organized as follows. Section~\ref{sec:chain-rings} reviews basic concepts on finite chain rings and linear algebra over them. Section~\ref{sec:motivation} motivates the study of MMCs over finite chain rings, while Section~\ref{sec:channel-model} formalizes the channel model. Section~\ref{sec:mmc-finite-field} reviews some of the existing results on MMCs over finite fields, and Section~\ref{sec:mmc-chain-ring} contains our contributions about MMCs over finite chain rings. Finally, Section~\ref{sec:conclusion} concludes the paper.

\section{Background on Finite Chain Rings} \label{sec:chain-rings}

We now present some basic results on finite chain rings and linear algebra over them. For more details, we refer the reader to~\cite{McDonald.74,Nechaev.08,Honold.Landjev.00,Brown.92}. By the term \emph{ring} we always mean a commutative ring with identity $1 \neq 0$.

\subsection{Finite Chain Rings}

A ring $R$ is called a \emph{chain ring} if, for any two ideals $I, J$ of $R$, either $I \subseteq J$ or $J \subseteq I$. It is known that a finite ring $R$ is a chain ring if and only if $R$ is both \emph{principal} (i.e., all of its ideals are generated by a single element) and \emph{local} (i.e., the ring has a unique maximal ideal). Let $\pi \in R$ be any generator for the maximal ideal of~$R$, and let $s$ be the nilpotency index of $\pi$ (i.e., the smallest integer $s$ such that $\pi^s = 0$). Then, $R$ has precisely $s + 1$ ideals, namely, \[ R = \langle \pi^0 \rangle \supset \langle \pi^1 \rangle \supset \cdots \supset \langle \pi^{s -1} \rangle \supset \langle \pi^s \rangle = \{ 0 \}, \] where $\langle x \rangle$ denotes the ideal generated by $x \in R$. Furthermore, it is also known that the quotient $R / \langle \pi \rangle$ is a field, called the \emph{residue field} of $R$. If $q = |R / \langle \pi \rangle|$, then the size of each ideal of $R$ is $|\langle \pi^i \rangle| = q^{s-i}$, for $0 \leq i \leq s$; in particular, $|R| = q^s$. Note that $s=1$ (so that $\pi = 0$) if and only if $R$ is a finite field.

In this paper, if $R$ is a finite chain ring with $s$ non-zero ideals and residue field of order~$q$, then we say that $R$ is a \emph{$(q, s)$ chain ring}. For instance, $\ZZ_8 = \{ 0, 1, \ldots, 7 \}$, the ring of integers modulo $8$, is a $(2, 3)$ chain ring. Its ideals are $\langle 1 \rangle = \ZZ_8$, $\langle 2 \rangle = \{ 0, 2, 4, 6 \}$, $\langle 4 \rangle = \{ 0, 4 \}$, and $\langle 0 \rangle = \{ 0 \}$, and its residue field is $\ZZ_8 / \langle 2 \rangle \cong \FF_2$. Note, however, that two $(q,s)$ chain rings need not be isomorphic.

Let $R$ be a $(q, s)$ chain ring. In addition, let $\pi \in R$ be a fixed generator for its maximal ideal, and let $\Gamma \subseteq R$ be a fixed set of coset representatives for the residue field $R / \langle \pi \rangle$. Without loss of generality, assume $0 \in \Gamma$.\footnote{A particularly nice, canonical choice for $\Gamma$ is $\Gamma(R) = \{ x \in R : x^q = x \}$, known as the \emph{Teichm\"uller coordinate set} of~$R$.} Then, every element $x \in R$ can be written uniquely~as \[ x = \sum_{i=0}^{s-1} x^{(i)} \pi^i, \] where $x^{(i)} \in \Gamma$, for $0 \leq i < s$. The above expression is known as the \emph{$\pi$-adic expansion of~$x$} (with respect to $\Gamma$). For example, the $2$-adic expansion of $6 \in \ZZ_8$ with respect to $\Gamma = \{0,1\}$ is $6 = 0 \cdot 2^0 + 1 \cdot 2^1 + 1 \cdot 2^2$, i.e., the standard binary expansion of $6$.

Note that the uniqueness of the $\pi$-adic expansion (given~$\Gamma$) allows us to define the maps $(\cdot)^{(i)} : R \to \Gamma$, for $0 \leq i < s$. We also define \[ x^{\underline{i}} = \sum_{j=0}^{i-1} x^{(j)} \pi^j, \] for $0 \leq i \leq s$. One can show that $x^{\underline{i}} \equiv_{\pi^i} x$ for all $x \in R$, where $\equiv_a$ denotes congruence modulo $a$ (i.e., $x \equiv_a y$ if and only if $x - y \in \langle a \rangle$). In particular, $x^{(0)} = x^{\underline{1}} \equiv_\pi x$.

\subsection{Modules over Finite Chain Rings}

An \emph{$s$-shape} $\mu = (\mu_0, \mu_1, \ldots, \mu_{s-1})$ is simply a non-decreasing sequence of $s$ non-negative integers, that is, $0 \leq \mu_0 \leq \mu_1 \leq \cdots \leq \mu_{s-1}$. For convenience, we may write the $s$-shape $(m, m, \ldots, m)$ simply as $m$. Also, we set $\mu_{-1} = 0$ whenever it appears on our expressions.

Let $\lambda$ and $\mu$ be two $s$-shapes. We write $\lambda \preceq \mu$ if $\lambda_i \leq \mu_i$ for $0 \leq i < s$; otherwise, we write $\lambda \npreceq \mu$. This yields a partial ordering on the set of all $s$-shapes. Note that, according to our convention, $\lambda \preceq m$ means $\lambda_i \leq m$ for $0 \leq i < s$.

We define the addition of $s$-shapes in a component-wise fashion, that is, $\mu + \lambda = (\mu_0 + \lambda_0, \mu_1 + \lambda_1, \ldots, \mu_{s-1} + \lambda_{s-1})$. The subtraction of $s$-shapes in a component-wise fashion is not always well-defined (because we can get negative elements, or a sequence which is not non-decreasing). But we define $\mu - n = (\mu_0 - n, \mu_1 - n, \ldots, \mu_{s-1} -n)$, provided $n \leq \mu_0$, and $n - \mu = (n - \mu_{s-1}, \ldots, n - \mu_1, n - \mu_0)$, provided $n \geq \mu_{s-1}$, which clearly are well-defined $s$-shapes. Finally, we set $|\mu| = \mu_0 + \mu_1 + \cdots + \mu_{s-1}$.

Let $\mu = (\mu_0, \mu_1, \ldots, \mu_{s-1})$ be an $s$-shape. We define \[ R^\mu \triangleq \underbrace{\langle 1 \rangle \times \cdots \times \langle 1 \rangle}_{\mu_0} \times \underbrace{\langle \pi \rangle \times \cdots \times \langle \pi \rangle}_{\mu_1 - \mu_0} \times \cdots \times \underbrace{\langle \pi^{s-1} \rangle \times \cdots \times \langle \pi^{s-1} \rangle}_{\mu_{s-1} - \mu_{s-2}}. \] Clearly, being a Cartesian product of ideals, $R^\mu$ is a finite $R$-module. Conversely, every finite $R$-module $U$ is isomorphic to $R^\mu$ for some \emph{unique} $s$-shape~$\mu$~\cite[Theorem 2.2]{Honold.Landjev.00}. We call $\mu$ the \emph{shape} of $U$, and write $\mu = \shape U$. Thus, two finite $R$-modules are isomorphic precisely when they have the same shape. Also, from the fact that the size of the ideal $\langle \pi^i \rangle$ is given by $q^{s-i}$, we conclude that \begin{equation} \label{eq:size-module} |R^\mu| = q^{|\mu|}. \end{equation}

Note that, according to our convention that $m = (m, m, \ldots, m)$, the notation $R^m$ stands for the same object, whether $m$ is interpreted as an integer or as an $s$-shape. Also, in the finite field case ($s=1$), modules are vector spaces, and we have $\shape U = (m)$, where $m$ is the vector space dimension of $U$.

\subsection{Matrices over Finite Chain Rings}

For any subset $S \subseteq R$, we denote by $S^{m \times n}$ the set of all $m \times n$ matrices with entries in~$S$. The set of all invertible $n \times n$ matrices over $R$ is called the \emph{general linear group of degree $n$ over $R$}, and is denoted by $\GL_n(R)$.

Let $A \in R^{m \times n}$, and set $r = \min \{n, m \}$. A diagonal matrix (not necessarily square) $$D = \diag(d_1, d_2, \ldots, d_r) \in R^{m \times n}$$ is called a \emph{Smith normal form} of $A$ if there exist matrices $P \in \GL_m(R)$ and $Q \in \GL_n(R)$ (not necessarily unique) such that $A = P D Q$ and $d_1 \mid d_2 \mid \cdots \mid d_r$. It is known that matrices over principal rings (in particular, finite chain rings) always have a Smith normal form, which is unique up to multiplication of the diagonal entries by units. In this work, we shall require such entries to be powers of $\pi \in R$; by doing so, the Smith normal form becomes (absolutely) unique.

Let $\row A$ and $\col A$ denote the row and column span of $A \in R^{m \times n}$, respectively. Clearly, $\row A$ and $\col A$ are both $R$-modules. Moreover, by using the Smith normal form, we can easily prove that $\row A$ is isomorphic to $\col A$. We define the \emph{shape} of $A$ as $\shape A = \shape (\row A) = \shape (\col A)$. We thus have that $\mu = \shape A$ if and only if the Smith normal form of $A$ is given by \begin{equation} \label{eq:smith-normal-form} \diag (\underbrace{1, \ldots, 1}_{\mu_0}, \underbrace{\pi, \ldots, \pi}_{\mu_1 - \mu_0}, \ldots, \underbrace{\pi^{s-1}, \ldots, \pi^{s-1}}_{\mu_{s-1} - \mu_{s-2}}, \underbrace{0, \ldots, 0}_{r - \mu_{s-1}}), \end{equation} where $r = \min \{n, m \}$. For example, consider the matrix \[ A = \mat{4 & 3 & 6 \\ 6 & 7 & 2} \] over $\ZZ_8$. Then, $A = P D Q$, where \[ P = \mat{1 & 0 \\ 1 & 1}, \quad D = \mat{1 & 0 & 0 \\ 0 & 2 & 0}, \quad Q = \mat{4 & 3 & 6 \\ 1 & 2 & 6 \\ 5 & 6 & 3}, \] so that $\shape A = (1, 2, 2)$. We also define the null space of~$A$ as usual, that is, $\nul A = \{ x \in R^n : Ax = 0 \}$. From the first isomorphism theorem~\cite[\S 10.2]{Dummit.04}, $\col A \cong R^n / \nul A$. Also,~\cite[Theorem~2.5]{Honold.Landjev.00} \begin{equation} \label{eq:rank-nullity} \shape A = n - \shape (\nul A). \end{equation}

\subsection{Matrices with Row Constraints} \label{sub:matrices-row-constraints}

Let $\lambda$ be an $s$-shape. We denote by $R^{n \times \lambda}$ the subset of matrices in $R^{n \times \ell}$ whose rows are elements of~$R^\lambda$, where $\ell = \lambda_{s-1}$. From \eqref{eq:size-module}, we have $|R^{n \times \lambda}| = q^{n |\lambda|}$. For instance, let $R = \ZZ_8$, $n = 2$, and $\lambda = (1,2,3)$, so that $\ell = 3$. Then, \[ R^{n \times \lambda} = \left\{ \mat{x_{11} & 2x_{12} & 4x_{13} \\ x_{21} & 2x_{22} & 4x_{23}} : x_{i,j} \in R \right\} \subseteq R^{n \times \ell}. \] Note that the matrix $A$ above does not belong to $R^{n \times \lambda}$, while $D$ does.

Finally, we extend the $\pi$-adic expansion map $(\cdot)^{(i)}$ to matrices over $R$ in an element-wise fashion. Thus, $A \in R^{n \times \lambda}$ if and only if $A^{(i)} = \mat { B_i & 0 } \in \Gamma^{n \times \ell}$, for some $B_i \in \Gamma^{n \times \lambda_i}$, for $0 \leq i < s$.

\section{Motivating Examples} \label{sec:motivation}

\subsection{MMCs as End-to-End Models for PNC}

\begin{figure} \centering \includegraphics[page=1]{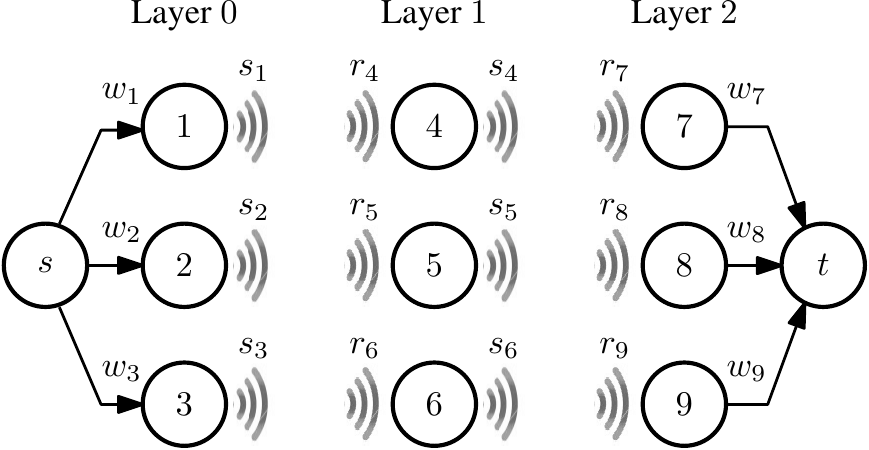} \caption{Wireless layered network with $L=3$ layers and $n=3$ relay nodes per layer.} \label{fig:pnc-network} \end{figure} Figure~\ref{fig:pnc-network} shows a wireless layered network with $L=3$ layers and $n=3$ relay nodes per layer. Suppose that the network employs physical-layer network coding, with the packets from the upper layer being elements of some $R$-module $W = R^\lambda$, where $R$ is a $(q,s)$ chain ring. Let $w_1, w_2, w_3 \in R^\lambda$ be the packets transmitted by the source node~$s$, and let $w_7, w_8, w_9 \in R^\lambda$ be the packets received by the sink node~$t$. Let $s_1, s_2, \ldots, s_6$ be the physical signals (complex vectors coming from a given lattice~\cite{Nazer.Gastpar.11,Feng.13}) transmitted by the nodes $1, 2, \ldots, 6$, respectively, and let $r_4, r_5, \ldots, r_9$ be the physical signals received by the nodes $4, 5, \ldots, 9$, respectively, as shown in the figure. Note that, in this example, for the sake of simplicity, the nodes $1$, $2$, and $3$ do not receive physical signals from node $s$, but rather packets $w_1, w_2, w_3$ coming directly from the upper layer. Similarly, the nodes $7$, $8$, and $9$ do not transmit physical signals to node $t$, but rather packets $w_7, w_8, w_9$ through the upper layer.

From Layer $0$ to Layer $1$, the system works as follows. Nodes $1$, $2$, and $3$ start by encoding the packets $w_1, w_2, w_3 \in R^\lambda$ into the signals $s_1, s_2, s_3$, respectively. The signals $s_1, s_2, s_3$ are then transmitted simultaneously, being subject to independent block fading and superimposed in the physical medium. Therefore, the signal received by node $j$, for $j = 4, 5, 6$, is given by $r_j = h_{1j} s_1 + h_{2j} s_2 + h_{3j} s_3 + n_j,$ where $h_{1j}, h_{2j}, h_{3j} \in \CC$ are fading coefficients and $n_j$ is a complex-valued noise vector. From $r_j$ and $(h_{1j}, h_{2j}, h_{3j})$, by employing the principles of PNC, the node $j$, for $j = 4, 5, 6$, can infer\footnote{Note that any additive error introduced at the physical layer may be avoided, at each relay node, by employing a linear error-detecting code over the underlying ring.} a linear combination $w_j \in R^\lambda$ of the packets $w_1, w_2, w_3$, that is, $w_j = b_{1j} w_1 + b_{2j} w_2 + b_{3j} w_3,$ for some $b_{1j}, b_{2j}, b_{3j} \in R$.

The system operates similarly from Layer $1$ to Layer $2$, so that, the node $j$, for $j = 7, 8, 9$, can infer a linear combination $w_j \in R^\lambda$ of the packets $w_4, w_5, w_6$, which is finally delivered to the sink node~$t$.

By $R$-module linearity, it is not hard to check that the relationship between the transmitted packets $X$ and the received packets $Y$, where \[ X = \mat{w_1 \\ w_2 \\ w_3} \in R^{n \times \lambda} \quad \text{and} \quad Y = \mat{w_7 \\ w_8 \\ w_9} \in R^{n \times \lambda}, \] is given by \[ Y = AX, \] where \[ A = \left[ \begin{array}{ccc} b_{47} & b_{57} & b_{67} \\ b_{48} & b_{58} & b_{68} \\ b_{49} & b_{59} & b_{69} \end{array} \right] \left[ \begin{array}{ccc} b_{14} & b_{24} & b_{34} \\ b_{15} & b_{25} & b_{35} \\ b_{16} & b_{26} & b_{36} \end{array} \right] \in R^{m \times n}. \] In other words, the end-to-end communication between the source node and the sink node is suitably modeled by an MMC over a finite chain ring.

\subsection{Communication via MMCs over Finite Chain Rings}

Consider now an MMC over the chain ring $R = \ZZ_8$ with packet space given by $W = \ZZ_8 \times 2\ZZ_8 = R^\lambda$, where $\lambda = (1,2,2)$. Assume that $n = m = 3$. Suppose that the receiver observes $(Y, A) \in R^{m \times \lambda} \times R^{m \times n}$, where \[ Y = \mat{ 7 & 2 \\ 4 & 4 \\ 6 & 0 \\ 4 & 0 }, \quad \text{and} \quad A = \mat{ 1 & 0 & 0 & 0 \\ 0 & 2 & 0 & 0 \\ 0 & 0 & 2 & 0 \\ 0 & 0 & 0 & 4 \\ }. \] What information can the receiver extract about the channel input $X = \mat{x_{ij}} \in R^{n \times \lambda}$ ($1 \leq i \leq 4$, $1 \leq j \leq 2$)? From the equation $AX = Y$ we may conclude that \begin{align*} \begin{cases} \phantom{1} x_{11} & = 7 \\ 2 x_{21} & = 4 \\ 2 x_{31} & = 6 \\ 4 x_{41} & = 4 \end{cases} \implies \begin{cases} x_{11} & = \framebox[1.25em][c]{1} \cdot 4 + \framebox[1.25em][c]{1} \cdot 2 + \framebox[1.25em][c]{1} \cdot 1 \\ x_{21} & = \makebox[1.25em][c]{?} \cdot 4 + \framebox[1.25em][c]{1} \cdot 2 + \framebox[1.25em][c]{0} \cdot 1 \\ x_{31} & = \makebox[1.25em][c]{?} \cdot 4 + \framebox[1.25em][c]{1} \cdot 2 + \framebox[1.25em][c]{1} \cdot 1 \\ x_{41} & = \makebox[1.25em][c]{?} \cdot 4 + \makebox[1.25em][c]{?} \cdot 2 + \framebox[1.25em][c]{1} \cdot 1 \end{cases} \end{align*} and \begin{align*} \begin{cases} \phantom{1} x_{12} & = 2 \\ 2 x_{22} & = 4 \\ 2 x_{32} & = 0 \\ 4 x_{42} & = 0 \end{cases} \implies \begin{cases} x_{12} & = \framebox[1.25em][c]{0} \cdot 4 + \framebox[1.25em][c]{1} \cdot 2 + {0} \cdot 1 \\ x_{22} & = \makebox[1.25em][c]{?} \cdot 4 + \framebox[1.25em][c]{1} \cdot 2 + {0} \cdot 1 \\ x_{32} & = \makebox[1.25em][c]{?} \cdot 4 + \framebox[1.25em][c]{0} \cdot 2 + {0} \cdot 1 \\ x_{42} & = \makebox[1.25em][c]{?} \cdot 4 + \makebox[1.25em][c]{?} \cdot 2 + {0} \cdot 1 \end{cases} \end{align*} where ``$?$'' denotes unknown entries, the squared entries indicates information that the receiver can extract about $X$, and the non-squared entries (forced to $0$) are due to the packet space constraints. Note that the unknown entries are due to $\rho = \shape A = (1,3,4)$, while the entries forced to $0$ are due to $\lambda = (1,2,2)$ (see~\S\ref{sub:matrices-row-constraints}).

Therefore, in the (non realistic) situation in which \emph{both} the transmitter and the receiver know the transfer matrix, it is clear that $4 + 6 + 2 = 12 \text{ bits}$ of information can be sent through the channel. (In general, it is not hard to check that $\rho_2 \lambda_0 + \rho_1 \lambda_1 + \rho_0 \lambda_2$ bits can be transmitted.) For such, the squared bits or $X$ should be set to information bits, while the remaining bits cannot carry information.

This idea can be generalized if $A$ is not diagonal, but an arbitrary matrix of shape $\rho$. In this case, we compute invertible matrices $P$ and $Q$ such that $A = PDQ$, where $D$ is the Smith normal form of $A$, as given by~\eqref{eq:smith-normal-form}. We then set $\tilde{Y} \triangleq P^{-1} Y$ and $\tilde{X} \triangleq QX$, so that we can communicate using the equivalent channel $\tilde{Y} = D \tilde{X}$ by employing the same scheme as before.

In this paper, we consider the problem of transmission of information through finite-chain-ring MMCs in the more realistic situation where the transfer matrix is unknown to the transmitter but known to the receiver (i.e., the coherent scenario) and chosen randomly according to some given probability distribution. It is shown that we can transmit the same amount of information as if the transmitter knew the transfer matrix, that is, at a rate given by $\EV[\bm{\rho}_2] \lambda_0 + \EV[\bm{\rho}_1] \lambda_1 + \EV[\bm{\rho}_0] \lambda_2$, where $\bm{\rho} = (\bm{\rho}_0, \bm{\rho}_1, \bm{\rho}_2)$ is the random variable representing the shape of the random transfer matrix, and $\EV[\cdot]$ denotes expected value. To do so, however, a non-trivial coding scheme (potentially using the channel multiple times and allowing a non-zero but vanishing probability of error) is needed. We also address the problem of reliable communication with a single use of the channel. In this case, we show that, as long as $\lambda_0 \geq n$ and the shape deficiency of the transfer matrix is at most a given value, say~$\beta$, we can have a one-shot zero-error coding scheme of rate given by $(n - \beta_0) \lambda_0 + (n - \beta_1) \lambda_1 + (n - \beta_2) \lambda_2$, which is the best rate one could achieve with zero error.

\section{Channel Model} \label{sec:channel-model}

We next formalize the channel model. Let $R$ be a $(q,s)$ chain ring, let $n$ and $m$ be positive integers, and let $\lambda$ be an $s$-shape. Also, let $p_{\bm{A}}$ be a probability distribution over $R^{m \times n}$. From these, we can define the \emph{coherent MMC over~$R$} as a \emph{discrete memoryless channel}~(see, e.g.,~\cite{Cover.06}) with \emph{input alphabet} $\calX = R^{n \times \lambda}$, \emph{output alphabet} $\calY = R^{m \times \lambda} \times R^{m \times n}$, and \emph{channel transition probability} \[ p_{\bm{Y},\bm{A}|\bm{X}}(Y,A|X) = \begin{cases} p_{\bm{A}}(A), & \text{if } Y = AX, \\ 0, & \text{otherwise}. \end{cases} \] In this work, we shall denote the channel just defined by $\CMMC (n, m, \lambda, p_{\bm{A}})$, with the dependence on $R$ being implicit. We also make use of the random variable ${\bm \rho} = \shape {\bm{A}}$, distributed according to \[ p_{\bm \rho}(\rho) = \sum_{A : \, \shape A = \rho} p_{\bm{A}}(A), \] Finally, we set $\ell = \lambda_{s-1}$ (interpreted as the packet length).

A \emph{matrix (block) code} of \emph{length} $N$ is defined by a pair $(\calC, \Phi)$, where $\calC \subseteq (R^{n \times \lambda})^N$ is called the \emph{codebook}, and $\Phi : (R^{m \times \lambda} \times R^{m \times n})^N \to \calC$ is called the \emph{decoding function}. We sometimes abuse the notation and write $\calC$ instead of $(\calC, \Phi)$. The \emph{rate} of the code $\calC$ is defined by $\mathrm{R}(\calC) = (\log |\calC|) / N$, and its \emph{probability of error} in the channel, denoted by $\mathrm{P}_{\mathrm{e}}(\calC)$, is defined as usual~\cite{Cover.06}. When $N=1$, we say that $\calC$ is a \emph{one-shot} code; otherwise, we say that $\calC$ is a \emph{multi-shot} code.

The \emph{capacity} of the channel is given by \[ C = \max_{p_{\bm{X}}} I(\bm{X} ; \bm{Y}, \bm{A}), \] where $I(\bm{X} ; \bm{Y}, \bm{A})$ is the mutual information between the input~$\bm{X}$ and the output~$(\bm{Y}, \bm{A})$, and the maximization is over all possible input distributions~$p_{\bm{X}}$.

From now on, all logarithms are to the base $q$, so that rates and capacities will always be expressed in $q$-ary digits (per channel use).

\section{Review of the MMC over a Finite Field} \label{sec:mmc-finite-field}

In this section, we briefly review some of the existing results about the coherent MMC over a finite field (i.e.,~$R = \FF_q$). Note that, in this case, $s = 1$, $\lambda = \ell$, and $\bm{\rho} = \rank{\bm{A}} \triangleq \bm{r}$.

\subsection{Finite-Field Coherent MMC}

The following result is due to Yang et al.~\cite{Yang.10.ISIT,Yang.10.arXiv}.

\begin{theorem} \label{thm:capacity-finite-field} \cite[Prop.~1]{Yang.10.arXiv} The capacity of $\CMMC(n, m, \ell, p_{\bm{A}})$ is given by \[ C = \EV[\bm{r}] \ell, \] and is achieved if the input is uniformly distributed over $\FF_q^{n \times \ell}$. In particular, the capacity depends on $p_{\bm{A}}$ only through $\EV[\bm{r}]$. \end{theorem}

Also in~\cite{Yang.10.ISIT,Yang.10.arXiv}, two multi-shot coding schemes for MMCs over finite fields are proposed, which are able to achieve the channel capacity given in Theorem~\ref{thm:capacity-finite-field}. The first scheme makes use of rank-distance codes (more on these later) and requires $\ell \geq n$ in order to be capacity-achieving; the second scheme is based on random coding and imposes no restriction on~$\ell$. Both schemes have polynomial time complexity. Also important, both coding schemes are ``universal'' in the sense that only the value of $\EV[\bm{r}]$ is taken into account in the code construction (the full knowledge of~$p_{\bm{A}}$, or even~$p_{\bm{r}}$, is not required).

\subsection{Rank Deficiency Correction Guarantees}

We say that a one-shot matrix code $\calC \subseteq \FF_q^{n \times \ell}$ is \emph{$b$-rank-deficiency-correcting} if it is possible to uniquely recover $X$ from $(Y, A)$, where $Y = AX$, as long as $X \in \calC$ and $\rank A \geq n - b$. In other words, $\calC$ is $b$-rank-deficiency-correcting if and only if, for every two distinct codewords $X_1, X_2 \in \calC$, there is no matrix $A \in \FF_q^{m \times n}$ such that $\rank A \geq n - b$ and $AX_1 = AX_2$.

Recall that the \emph{rank distance} between two matrices $X_1, X_2 \in \FF_q^{n \times \ell}$ is defined as $\dR(X_1, X_2) = \rank(X_2 - X_1)$. For a code $\calC \subseteq \FF_q^{n \times \ell}$, define $\dR(\calC) = \min \{ \dR(X_1, X_2) : X_1, X_2 \in \calC, X_1 \neq X_2 \}$, called the \emph{minimum distance} of the code. The rank distance provides a necessary and sufficient condition under which a code is $b$-rank-deficiency-correcting. The following result is a special case of a result due to Silva et al.~\cite{Silva.08,Silva.Kschischang.11}.

\begin{theorem} \label{thm:iff-finite-field} \cite[Thm.~2]{Silva.Kschischang.11} A code $\calC \subseteq \FF_q^{n \times \ell}$ is $b$-rank-deficiency-correcting if and only if $\dR(\calC) > b$. \end{theorem}

Rank-distance codes were studied by Gabidulin~\cite{Gabidulin.85}, which shows that any linear rank-distance code $\calC \subseteq \FF_q^{n \times \ell}$ of dimension~$k$ has rate given by \[ \Rate(\calC) = k \ell \] and minimum distance satisfying \[ \dR(\calC) \leq n - k + 1. \] Codes achieving equality in the above are said to be \emph{maximum rank distance} (MRD) codes. A class of such codes for every $n$, $\ell$, $k$, and $q$ such that $\ell \geq n$ was presented by Gabidulin. Theorem~\ref{thm:iff-finite-field} implies that any linear MRD code of dimension $k$ is $(n - k)$-rank-deficiency-correcting.

Finally, note that if a code $\calC \subseteq \FF_q^{n \times \ell}$ is $(n - r)$-rank-deficiency-correcting for every~$r$ in the support of $\bm{r} = \rank \bm{A}$, then $\calC$ has $\Perr(\calC) = 0$ in $\CMMC(n, m, \ell, p_{\bm{A}})$. In particular, if $\bm{r}$ is a constant, a zero-error capacity-achieving coding scheme can be obtained by employing a linear MRD code of dimension $k = r$.

\section{The MMC over a Finite Chain Ring} \label{sec:mmc-chain-ring}

This section contains the contributions of the paper, where we consider again the case of a general $(q, s)$ chain ring~$R$.

\subsection{Channel Capacity}

We start by computing the channel capacity. The following result generalizes Theorem~\ref{thm:capacity-finite-field}.

\begin{theorem} \label{thm:capacity} The capacity of $\CMMC(n, m, \lambda, p_{\bm{A}})$ is given by \[ C = \sum_{i=0}^{s-1} \EV [\bm{\rho}_{s-i-1}] \lambda_i, \] and is achieved if the input is uniformly distributed over $R^{n \times \lambda}$. In particular, the capacity depends on $p_{\bm{A}}$ only through $\EV[\bm{\rho}]$. \end{theorem}

The following example illustrates the theorem.

\medskip

\begin{example} \begin{figure*} \centering \begin{subfigure}{0.48\textwidth} \centering \includegraphics[width=8cm]{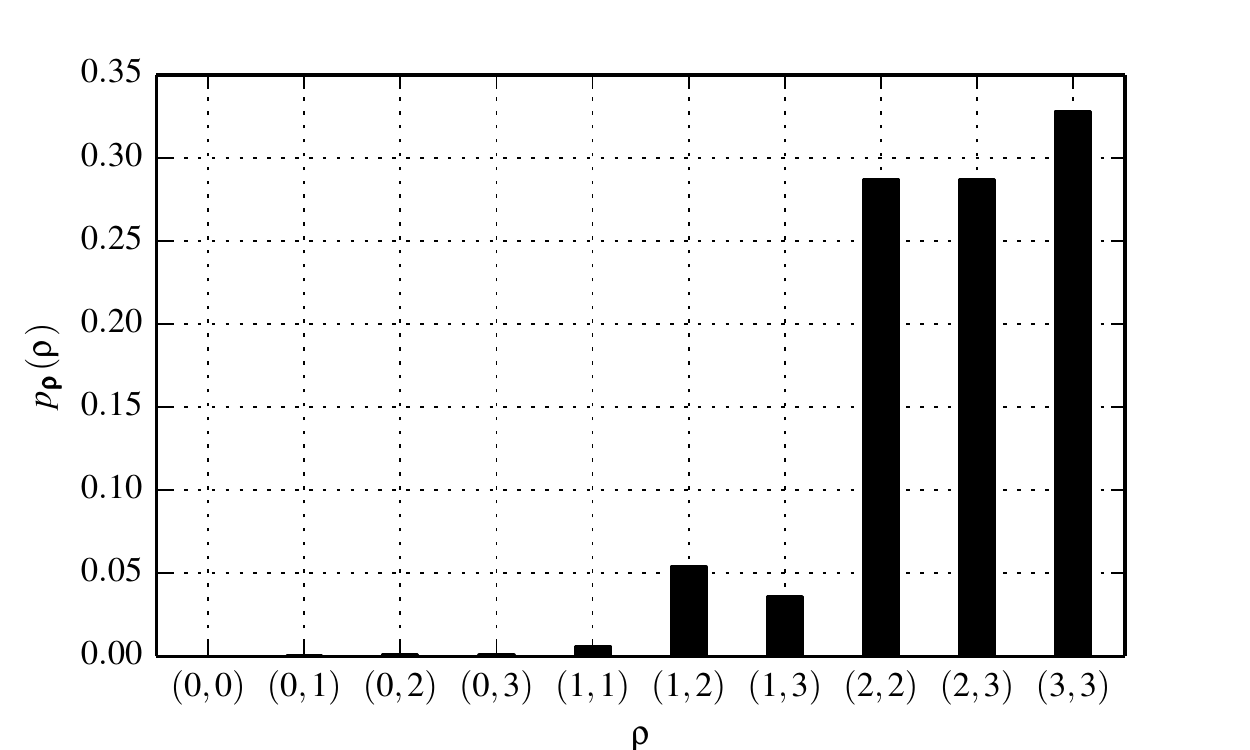} \caption{} \label{fig:example-shape-dist} \end{subfigure}

\begin{subfigure}{0.48\textwidth} \centering \includegraphics[width=8cm]{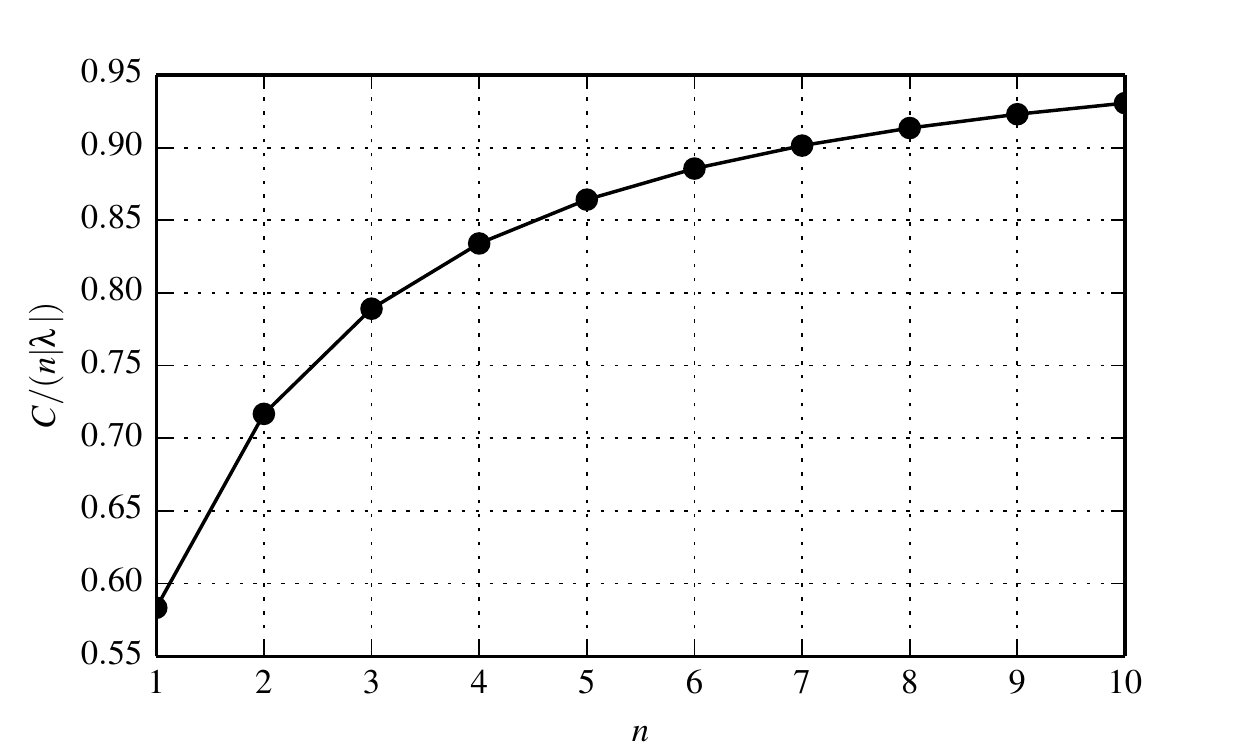} \caption{} \label{fig:example-capacity-n} \end{subfigure} \quad \begin{subfigure}{0.48\textwidth} \centering \includegraphics[width=8cm]{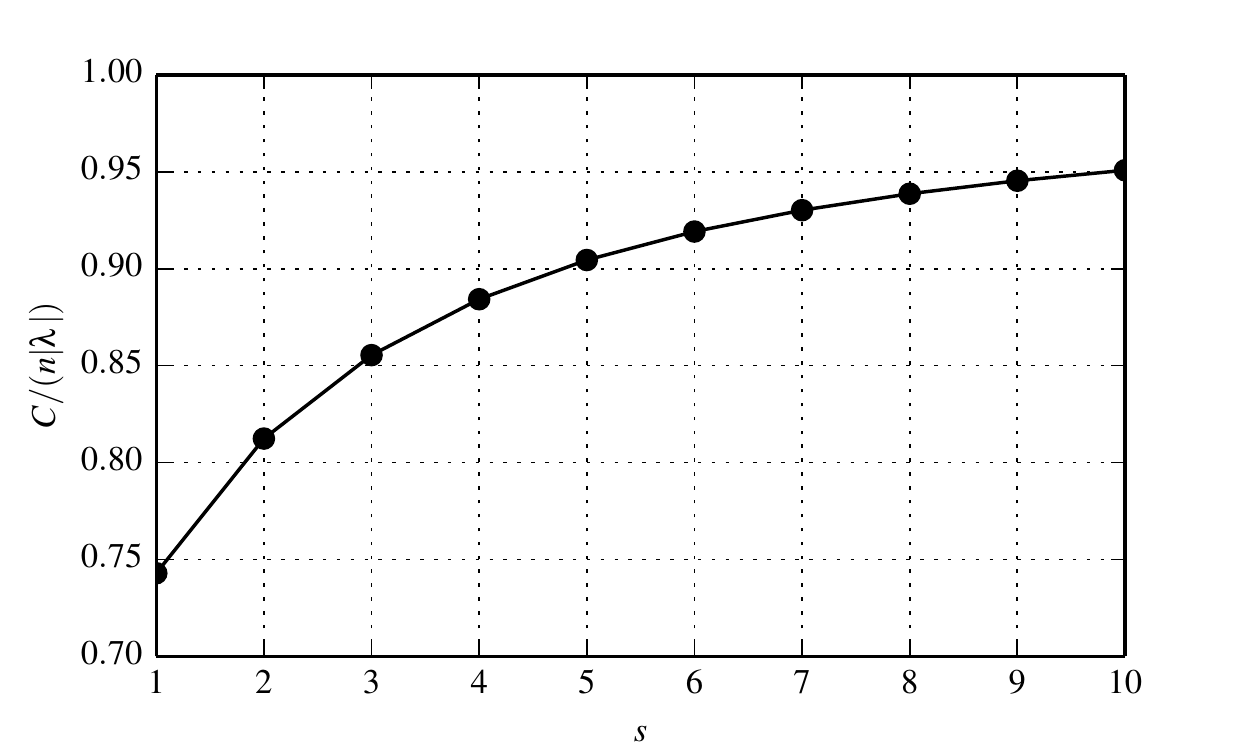} \caption{} \label{fig:example-capacity-s} \end{subfigure} \caption{(a) Shape distribution for $n = 3$ and $s = 2$. (b) Channel capacity (normalized by $n|\lambda|$) as a function of $n$, for $s = 2$ and $\lambda = (\lambda_0,2\lambda_0)$. (c) Channel capacity (normalized by $n |\lambda|$) as a function of $s$, for $n = 3$ and $\lambda = \ell$.} \label{fig:example} \end{figure*} Let $R = \ZZ_{2^s}$, which is a $(2, s)$ chain ring. In addition, suppose that the transfer matrix~$\bm{A} \in R^{m \times n}$ has i.i.d.\ entries uniform over $R$, which is equivalent to say that $\bm{A}$ is uniformly distributed over $R^{m \times n}$ (this is analogous to the transfer matrix distribution considered in~\cite{Jafari.11}). Therefore, the shape distribution of the transfer matrix can be expressed as \[ p_{\bm{\rho}}(\rho) = \frac{|\calT_{\rho}(R^{m \times n})|}{|R^{m \times n}|}, \] where $\calT_{\rho}(R^{m \times n})$ denotes the set of matrices in $R^{m \times n}$ whose shape is $\rho$ (its cardinality can be found in~\cite[Thm.~3]{Feng.13.arXiv}). Suppose, for simplicity that $n = m$. Figure~\ref{fig:example-shape-dist} shows the probability distribution of~$\bm{\rho}$ when $n = 3$ and $s = 2$. Figure~\ref{fig:example-capacity-n} shows the channel capacity, normalized by $n |\lambda|$, as a function of $n$, for $s = 2$ and packet space $W = R^\lambda$, where $\lambda = (\lambda_0, 2\lambda_0)$. Figure~\ref{fig:example-capacity-s} shows the normalized channel capacity as a function of $s$, for $n = 3$ and packet space $W = R^\ell$. \end{example}

\medskip

In order to prove Theorem~\ref{thm:capacity}, we need the following lemma.

\begin{lemma} \label{lem:entropy-GX} Let $\bm{X} \in R^{n \times \lambda}$ be a random matrix, let $A \in R^{m \times n}$ be any fixed matrix, and let $\rho = \shape A$. Define $\bm{Y} = A \bm{X} \in R^{m \times \lambda}$. Then, \[ H(\bm{Y}) \leq \sum_{i=0}^{s-1} \rho_{s-i-1} \lambda_i, \] where equality holds if $\bm{X}$ is uniformly distributed over $R^{n \times \lambda}$. \end{lemma}

\begin{proof} Note that $\bm{X}$ and $\bm{Y}$ can be expressed as \begin{align*} \bm{X} & = \mat{{\bm{X}}_0 & {\bm{X}}_1 & \cdots & {\bm{X}}_{s-1}}, \\ \bm{Y} & = \mat{{\bm{Y}}_0 & {\bm{Y}}_1 & \cdots & {\bm{Y}}_{s-1}}, \end{align*} where ${\bm{X}}_i \in \langle \pi^i \rangle^{n \times (\lambda_i - \lambda_{i-1})}$ and ${\bm{Y}}_i \in \langle \pi^i \rangle^{m \times (\lambda_i - \lambda_{i-1})}$, for $0 \leq i < s$. We have \[ {\bm{Y}}_i = A {\bm{X}}_i, \] so that the support of each of the columns of ${\bm{Y}}_i$ is a subset of~$\col \pi^i A$. We have $\shape \pi^i A = (0, \ldots, 0, \rho_0, \ldots, \rho_{s-i-1})$, so that, from~\eqref{eq:size-module}, we have $|\col \pi^i A| = q^{\rho_0 + \cdots + \rho_{s-i-1}}$. Therefore, the support of $\bm{Y}$ has size at most \begin{align*} \prod_{i=0}^{s-1} |\col \pi^i A|^{\lambda_i - \lambda_{i-1}} & = \prod_{i=0}^{s-1} q^{(\rho_0 + \cdots + \rho_{s-i-1})(\lambda_i - \lambda_{i-1})} \\ & = q^{\sum_{i=0}^{s-1} \rho_{s-i-1} \lambda_i}, \end{align*} from which the inequality follows.

Now suppose ${\bm{X}}$ is uniformly distributed over $R^{n \times \lambda}$. This means that ${\bm{X}}_i$ is uniformly distributed over $\langle \pi^i \rangle^{n \times (\lambda_i - \lambda_{i-1})}$. One may show that there exists ${\bm{X}}'_i$ uniformly distributed over $R^{n \times (\lambda_i - \lambda_{i-1})}$ such that ${\bm{X}}_i = \pi^i {\bm{X}}'_i$. Let ${\bm{y}}$ denote a column of ${\bm{Y}}_i$, whose support is $\col \pi^i A$. Since ${\bm{Y}}_i = A {\bm{X}}_i = \pi^i A {\bm{X}}'_i$, we have, for every $y \in \col \pi^i A$, \begin{align*} \Pr[\bm{y} = y] & = \frac {|\{ x' \in R^n : \pi^i A x' = y \}|} {|R^n|} \\ & = \frac {|\nul \pi^i A|} {|R^n|} \\ & = \frac {1} {|\col \pi^i A|}, \end{align*} that is, $\bm{y}$ is uniformly distributed over its support. Therefore, $\bm{Y}$ itself is also uniformly distributed over its support. This concludes the proof. \end{proof}

We can now prove Theorem~\ref{thm:capacity}.

\begin{proof}[Proof of Theorem~\ref{thm:capacity}] The channel mutual information is given by \begin{align*} I(\bm{X} ; \bm{Y}, \bm{A}) & = I(\bm{X} ; \bm{Y} | \bm{A}) + I(\bm{X} ; \bm{A}) \\ & = H(\bm{Y} | \bm{A}) - H(\bm{Y} | \bm{X} , \bm{A}) + I(\bm{X} ; \bm{A}) \\ & = H(\bm{Y} | \bm{A}), \end{align*} where $H(\bm{Y} | \bm{X}, \bm{A}) = 0$ since $\bm{Y} = \bm{A} \bm{X}$, and $I(\bm{X} ; \bm{A}) = 0$ since $\bm{X}$ and $\bm{A}$ are independent. Thus, \begin{equation*} I(\bm{X} ; \bm{Y}, \bm{A}) = H(\bm{Y} | \bm{A}) = \sum_A p_{\bm{A}}(A) H(\bm{Y} | \bm{A} = A), \end{equation*} and the result follows from Lemma~\ref{lem:entropy-GX}. \end{proof}

\subsection{Coding Scheme}

Here we describe the proposed coding scheme. Before doing so, we present two simple lemmas regarding the solution of systems of linear equations over a finite chain ring, via the $\pi$-adic expansion. These results will serve as a basis for the coding scheme. From now on, let $F = R / \langle \pi \rangle \cong \FF_q$.

\medskip \subsubsection{Auxiliary Results}

The first problem turns a system of linear equations over the chain ring into multiple systems over the residue field.

\begin{lemma} \label{lem:linalg-1} Let $y \in R^n$ and $A \in \GL_n(R)$. Let $x \in R^n$ be the (unique) solution of $Ax = y$. Then, the $\pi$-adic expansion of $x$ can be obtained recursively from \[ A^{(0)} x^{(i)} \equiv_\pi y^{(i)} - \big( A x^{\underline{i}} \big)^{(i)}, \] for $0 \leq i < s$. \end{lemma}

\begin{proof} For $0 \leq i < s$, we have \[ y = Ax = A \sum_{j=0}^{i-1} x^{(j)} \pi^j + A x^{(i)} \pi^i + A \sum_{j=i+1}^{s-1} x^{(j)} \pi^j, \] so that, from Lemma~\ref{lem:pi-adic-aux}, \[ y^{(i)} \equiv_\pi \big( A x^{\underline{i}} \big)^{(i)} + \big( A x^{(i)} \big)^{(0)}. \] After simplifying and rearranging we get the equation displayed on the lemma. Since $A^{(0)} \in \GL_n(F)$, we can compute, recursively, $x^{(0)}, x^{(1)}, \ldots, x^{(s-1)}$. \end{proof}

The second problem deals with the solution of diagonal systems of linear equations. Let $M_{j:j'}$ denote the sub-matrix of $M$ consisting of rows~$j$ up to, \emph{but not including}, $j'$, where we index the matrix entries starting from~$0$.

\begin{lemma} \label{lem:linalg-2} Let $Y \in R^{m \times \lambda}$ and $D \in R^{m \times n}$, where $D$ is the Smith normal form of itself and has shape $\rho$. If $Y = DX$, then \begin{equation*} X^{(i)}_{0:\rho_{s-i-1}} = \mat { Y^{(i)}_{0 : \rho_0} \\ Y^{(i+1)}_{\rho_0 : \rho_1} \\ \vdots \\ Y^{(i+s-1)}_{\rho_{s-i-2} : \rho_{s-i-1}}}, \end{equation*} for $0 \leq i < s$. \end{lemma}

\begin{proof} Note that $Y = DX$ is equivalent to \begin{align*} Y_{0 : \rho_0} & = X_{0 : \rho_0}, \\ Y_{\rho_0 : \rho_1} & = \pi X_{\rho_0 : \rho_1}, \\ & \vdots \\ Y_{\rho_{s-2}: \rho_{s-1}} & = \pi^{s-1} X_{\rho_{s-2} : \rho_{s-1}}. \end{align*} From Lemma~\ref{lem:pi-adic-aux}, this implies \begin{alignat*}{2} X^{(i)}_{0 : \rho_0} & = Y^{(i)}_{0 : \rho_0}, && \quad 0 \leq i < s, \\ X^{(i)}_{\rho_0 : \rho_1}& = Y^{(i+1)}_{\rho_0 : \rho_1}, && \quad 0 \leq i < s-1, \\ & \ \, \vdots && \quad \qquad \vdots \\ X^{(i)}_{\rho_{s-2} : \rho_{s-1}} & = Y^{(i+s-1)}_{\rho_{s-2} : \rho_{s-1}}, && \quad 0 \leq i < 1, \end{alignat*} from which the result follows. \end{proof}

We are finally ready to present the coding scheme, which is based on the ideas of the two previous lemmas. For simplicity of exposition, we first address the particular case of one-shot codes. The general case will be discussed afterwards.

\medskip \subsubsection{Codebook}

We start with the codebook construction. Let $\calC_0, \calC_1, \ldots, \calC_{s-1}$, where $\calC_i \subseteq F^{n \times \lambda_i}$, for $0 \leq i < s$, be a sequence of one-shot matrix codes over the residue field~$F$. We will combine these component codes to obtain a matrix code $\calC \subseteq R^{n \times \lambda}$ over the chain ring~$R$. We refer to $\calC_0, \calC_1, \ldots, \calC_{s-1}$ to as the \emph{component codes}, and to $\calC$ as the \emph{composite code}.

Denote by $\varphi : R \to F$ the natural projection map from $R$ onto $F$. Also, denote by $\bar{\varphi} : F \to \Gamma$ the coset representative selector map, with the property that $\varphi(\bar{\varphi}(x)) = x$ for all $x \in F$. The codebook $\calC \subseteq R^{n \times \lambda}$ is defined by \[ \calC = \left\{ \sum_{i=0}^{s-1} X^{(i)} \pi^i : X_i \in \calC_i, 0 \leq i < s \right\}, \] where \begin{equation} \label{eq:encoding-1} X^{(i)} = \mat{\bar{\varphi}(X_i) & 0} \in \Gamma^{n \times \ell}. \end{equation} It should be clear that the codewords in $\calC$ indeed satisfy the row constraints of $R^{n \times \lambda}$ (see~\S\ref{sub:matrices-row-constraints}). In addition, from the uniqueness of the $\pi$-adic expansion, \begin{equation} \label{eq:composite-size} \Rate(\calC) = \Rate(\calC_0) + \Rate(\calC_1) + \cdots + \Rate(\calC_{s-1}). \end{equation}

\medskip \subsubsection{Decoding}

We now describe the decoding procedure. Intuitively, the decoder decomposes a single MMC over the chain ring into multiple MMCs over the residue field. In the following, $M_{j \times k}$ denotes the upper-left $j \times k$ sub-matrix of~$M$.

\smallskip \textit{Step 1.} The decoder, which knows the transfer matrix~$A$, starts by computing its Smith normal form~$D \in R^{m \times n}$. It also computes $P \in \GL_m(R)$ and $Q \in \GL_n(R)$ such that $A = P D Q$.

\smallskip \textit{Step 2.} Let $\rho = \shape A = \shape D$. Define $\tilde{X} \triangleq Q X \in R^{n \times \lambda}$ (which is unknown to the receiver) and $\tilde{Y} \triangleq P^{-1} Y \in R^{m \times \lambda}$ (which is calculated at the receiver), so that $Y = A X$ is equivalent to \[ \tilde{Y} = D \tilde{X}. \] From this equation, the decoder can obtain partial information about $\tilde{X}$. More precisely, it can compute $\tilde{X}^{(i)}_{\rho_{s-i-1} \times \lambda_i}$, for $0 \leq i < s$, according to Lemma~\ref{lem:linalg-2}.

\smallskip \textit{Step 3.} In possession of $\tilde{X}^{(i)}_{\rho_{s-i-1} \times \lambda_i}$, for $0 \leq i < s$, the decoder will then try to decode $X$ based on the equation \[ \tilde{X} = Q X, \] in a \emph{multistage fashion}. Indeed, similarly to Lemma~\ref{lem:linalg-1}, we have, for $0 \leq i < s$, \[ \tilde{X}^{(i)} - \big( Q X^{\underline{i}} \big)^{(i)} \equiv_\pi Q^{(0)} X^{(i)}. \] Considering only the $\rho_{s-i-1}$ topmost rows (since the remaining rows are unknown), and keeping only the $\lambda_i$ leftmost columns (since the remaining columns are already known to be zero), we get \[ \tilde{X}^{(i)}_{\rho_{s-i-1} \times \lambda_i} - \big( Q_{\rho_{s-i-1} \times n}^{} X^{\underline{i}}_{n \times \lambda_i} \big)^{(i)} \equiv_\pi Q^{(0)}_{\rho_{s-i-1} \times n} X^{(i)}_{n \times \lambda_i}. \] Finally, projecting into~$F$ (that is, applying $\varphi$ to both sides), and appending enough zero rows (in order to obtain an $m \times n$ system) gives \begin{equation} \label{eq:layered-system-equations} Y_i = A_i X_i, \end{equation} where $Y_i \in F^{m \times \lambda_i}$ and $A_i \in F^{m \times n}$ are defined by \begin{equation} \label{eq:layer-Y} Y_i = \mat { \varphi \big(\tilde{X}^{(i)}_{\rho_{s-i-1} \times \lambda_i} \big) - \varphi \Big( \big( Q_{\rho_{s-i-1} \times n}^{} X^{\underline{i}}_{n \times \lambda_i} \big)^{(i)} \Big) \\ 0 }, \end{equation} and \begin{equation} \label{eq:layer-A} A_i = \mat{ \varphi \big( Q_{\rho_{s-i-1} \times n} \big) \\ 0}. \end{equation} Note that $Y_i$ can only be calculated after $X_0, X_1, \ldots, X_{i-1}$ are known. Therefore, in this step the decoder obtains, successively, estimates of $X_0, X_1, \ldots, X_{s-1}$ from~\eqref{eq:layered-system-equations}. Finally, it computes an estimate of $X$ according to~\eqref{eq:encoding-1} and the $\pi$-adic expansion.

\medskip \subsubsection{Extension to the Multi-Shot Case}

We finally consider the multi-shot case. Let $\calC_0, \calC_1, \ldots, \calC_{s-1}$ be a sequence of $N$-shot matrix codes (the component codes), where $\calC_i \subseteq (F^{n \times \lambda_i})^N$, for $0 \leq i < s$. The codewords of the composite code~$\calC$ are then given by $(X(1), X(2), \ldots, X(N)) \in (R^{n \times \lambda})^N$, where $X(j)$ is obtained from the $j$-th coordinates of the codewords of the component codes, similarly to the one-shot case.

Proceeding similarly to Steps 1 and 2 above, the decoder obtains $\tilde{X}^{(i)}_{\rho_{s-i-1} \times \lambda_i}(j)$, for $0 \leq i < s$ and $j = 1, \ldots, N$, and $Q(j)$, for $j = 1, \ldots, N$. Step 3 is also similar, with the important detail that the whole sequence $(X_i(1), X_i(2), \ldots, X_i(N)) \in \calC_i$ is decoded from $(Y_i(1), Y_i(2), \ldots, Y_i(N))$ and $(A_i(1), A_i(2), \ldots, A_i(N))$ by using the decoder of $\calC_i$, before proceeding to stage $i+1$.

\medskip \subsubsection{Computational Complexity}

The computational complexity of the scheme is simply the sum of the individual computational complexities of each component code, plus the cost of calculating the Smith normal form of $A$ (which can be done with $O(nm\min \{n , m \})$ operations in $R$), the cost of calculating $\tilde{Y}$ (taking $O(m^2(m + \ell))$ operations), and the cost of $s-1$ matrix multiplications and additions in~\eqref{eq:layer-Y} (taking $O(n^2\ell)$ operations each). As a consequence, if each component code has polynomial time complexity, then the composite code will also have polynomial time complexity.

\subsection{Achieving the Channel Capacity}

From the proposed coding scheme, it is now clear that the $i$-th component code $\calC_i$ should be aimed at $\CMMC(n, m, \allowbreak \lambda_i, p_{{\bm{A}}_i})$, where ${\bm{A}}_i \in F^{m \times n}$ is defined in~\eqref{eq:layer-A}. In principle, we could compute the probability distribution of~${\bm{A}}_i$, provided we have access to the probability distribution of~$\bm{A}$. Nevertheless, if we employ a universal coding scheme (see Section~\ref{sec:mmc-finite-field}), then the particular probability distribution of~${\bm{A}}_i$ becomes unimportant once we know the expected value of its rank. From~\eqref{eq:layer-A}, we have $\rank {\bm{A}}_i = {\bm \rho}_{s-i-1}$, so that, in this case, only the knowledge of $\EV[\bm{\rho}]$ is needed. Thus, the proposed coding scheme is ``universal'', provided the component codes are also universal. We next show that the scheme is able to achieve the channel capacity.

\begin{proposition} Let $\calC_i \subseteq F^{n \times \lambda_i}$ be a capacity-achieving~code in $\CMMC(n, m, \lambda_i, p_{{\bm{A}}_i})$, for $0 \leq i < s$, where ${\bm{A}}_i \in F^{m \times n}$ is defined in~\eqref{eq:layer-A}. Let $\calC \subseteq R^{n \times \lambda}$ be the composite code obtained from $\calC_0, \calC_1, \ldots, \calC_{s-1}$. Then, $\calC$ is a capacity-achieving code in $\CMMC(n, m, \lambda, p_{\bm{A}})$. \end{proposition}

\begin{proof} Since each $\calC_i$ is capacity-achieving in $\CMMC(n, m, \allowbreak \lambda_i, p_{{\bm{A}}_i})$, and since $\rank {\bm{A}}_i = {\bm \rho}_{s-i-1}$ [see~\eqref{eq:layer-A}], we have $\Rate(\calC_i)$ arbitrarily close to $\EV[{\bm \rho}_{s-i-1}] \lambda_i$. Thus, from~\eqref{eq:composite-size}, we have $\Rate(\calC)$ arbitrarily close to $\sum_i \EV[{\bm \rho}_{s-i-1}] \lambda_i$, which is the channel capacity. Now, from the union bound, the probability of error of $\calC$ in $\CMMC(n, m, \lambda, p_{\bm{A}})$ is upper-bounded by \begin{equation*} \Perr(\calC) \leq \Perr(\calC_0) + \Perr(\calC_1) + \cdots + \Perr(\calC_{s-1}), \end{equation*} where $\Perr(\calC_i)$ is the probability of error of $\calC_i$ in $\CMMC(n, m, \allowbreak \lambda_i, p_{\bm{A}_i})$. Since each $\calC_i$ is capacity-achieving, we have $\Perr(\calC_i)$ arbitrarily close to zero. Therefore, $\Perr(\calC)$ is also arbitrarily close to zero. \end{proof}

Recall that the two coding schemes proposed in~\cite{Yang.10.arXiv} (see Section~\ref{sec:mmc-finite-field}) are universal and have polynomial time complexity. Consequently, by using them as component codes, we can obtain a universal, capacity-achieving composite code with polynomial time complexity.

\subsection{One-Shot Reliable Communication}

Our last result is concerned with codes that guarantee reliable communication with a single use of the MMC, supposing that the ``(row) shape deficiency'' of the transfer matrix is bounded by a given value. In this paper, a one-shot matrix code $\calC \subseteq R^{n \times \lambda}$ is said to be \emph{$\beta$-shape-deficiency-correcting} if it is possible to uniquely recover $X$ from $(Y, A)$, where $Y = AX$, as long as $X \in \calC$ and $\shape A \succeq n - \beta$. In other words, $\calC$ is $b$-rank-deficiency-correcting if and only if, for every two distinct codewords $X_1, X_2 \in \calC$, there is no matrix $A \in R^{m \times n}$ such that $\shape A \succeq n - \beta$ and $AX_1 = AX_2$. The following result generalizes Theorem~\ref{thm:iff-finite-field}.

\begin{theorem} \label{thm:iff} A code $\calC \subseteq R^{n \times \lambda}$ is $\beta$-shape-deficiency-correcting if and only if there are no distinct $X_1, X_2 \in \calC$ such that $\shape(X_2 - X_1) \preceq \beta$. \end{theorem}

\begin{proof} Assume first that $\calC \subseteq R^{n \times \lambda}$ is $\beta$-shape-deficiency-correcting. Suppose, for the sake of contradiction, that there exist distinct $X_1, X_2 \in \calC$ such that $\shape(X_2 - X_1) \preceq \beta$. Let $A \in R^{m \times n}$ be any matrix such that $\row A = \nul (X_2 - X_1)^\tr$. Then, $A (X_2 - X_1) = 0$ so that $A X_1 = A X_2$. Also, \[ \shape A = \shape \nul (X_2 - X_1)^\tr = n - \shape (X_2 - X_1) \succeq n - \beta, \] where we made use of~\eqref{eq:rank-nullity}. This is a contradiction.

Assume now that there are no distinct $X_1, X_2 \in \calC$ such that $\shape(X_2 - X_1) \preceq \beta$. Suppose, for the sake of contradiction, that $\calC \subseteq R^{n \times \lambda}$ is $\beta$-shape-deficiency-correcting. Then, there exist distinct $X_1, X_2 \in \calC$ and a matrix $A \in R^{m \times n}$ such that $A X_1 = A X_2$ and $\shape A \succeq n - \beta$. We have $A(X_2 - X_1) = 0$, so that $\col (X_2 - X_1)$ must be a submodule of $\nul A$. Thus, \begin{equation*} \shape (X_2 - X_1) \preceq \shape (\nul A) = n - \shape A \preceq \beta, \end{equation*} where we again made use of~\eqref{eq:rank-nullity}. This is a contradiction. \end{proof}

We next show that the coding scheme proposed by this work can also provide shape deficiency correction guarantees. For such, the component codes are chosen to be MRD codes with suitable dimensions.

\begin{proposition} Suppose $\lambda_0 \geq n$. Let $\calC_i \subseteq F^{n \times \lambda_i}$ be a linear MRD code of dimension $n - \beta_i$, for $0 \leq i < s$. Let $\calC \subseteq R^{n \times \lambda}$ be the composite code obtained from $\calC_0, \calC_1, \ldots, \calC_{s-1}$. Then, $\Rate(\calC) = \sum_i (n - \beta_i) \lambda_i$, and $\calC$ is $\beta$-shape-deficiency-correcting. \end{proposition}

\begin{proof} We have $\Rate(\calC_i) = (n - \beta_i) \lambda_i$, so that the expression for~$\Rate(\calC)$ follows from~\eqref{eq:composite-size}. We now show that $\calC$ is $\beta$-shape-deficiency-correcting. Suppose not. Then, according to Theorem~\ref{thm:iff}, there exists two distinct codewords $X_1, X_2$ such that $\delta = \shape(X_2 - X_1) \preceq \beta$. On the other hand, we have $X_1 = \sum_{j=0}^{s-1} \bar{\varphi}(X_{1,j}) \pi^j$, for some $X_{1,j} \in \calC_j$, and likewise $X_2 = \sum_{j=0}^{s-1} \bar{\varphi}(X_{2,j}) \pi^j$, for some $X_{2,j} \in \calC_j$. Let $i$ such that $0 \leq i < s$ be the smallest integer satisfying $X_{1,i} \neq X_{2,i}$. We then have \[ X_2 - X_1 = \sum_{j=0}^{s-1} \bar{\varphi}(X_{2,j} - X_{1,j}) \pi^j = \sum_{j=i}^{s-1} \bar{\varphi}(X_{2,j} - X_{1,j}) \pi^j = \pi^i \sum_{j=0}^{s-i-1} \bar{\varphi}(X_{2,j+i} - X_{1,j+i}) \pi^j. \] From Lemma~\ref{lem:pi-shape} of Appendix~\ref{sec:auxiliary}, and from the fact that the $0$-th entry of $\shape A$ is $\rank \varphi(A)$, we conclude that \[ \delta_i = \rank(X_{2,i} - X_{1,i}) = \dR(X_{1,i}, X_{2,i}) \geq \dR(\calC_i) = \beta_i + 1 > \beta_i, \] where we also used the fact that $\calC_i$ is MRD. This contradicts the fact that $\delta = \shape(X_2 - X_1) \preceq \beta$, so that $\calC$ must be $\beta$-shape-deficiency-correcting. \end{proof}

Similarly to the finite-field case, if $\calC \subseteq R^{n \times \lambda}$ is $(n - \rho)$-shape-deficiency-correcting for every $\rho$ in the support of $\bm{\rho} = \shape \bm{A}$, then $\calC$ is a zero-error coding scheme for $\CMMC(n, m, \lambda, p_{\bm{A}})$. In particular, if the channel is such that $\bm{\rho} = \rho$ is a constant, the above construction yields a one-shot zero-error capacity-achieving code whose encoding and decoding procedures have polynomial time complexity.

\subsection{Extension to the Non-Coherent Scenario}

So far, we have only considered the coherent scenario, in which the instances of the transfer matrix are available to the receiver. Nevertheless, we can reuse the coding scheme proposed in this work even in a non-coherent scenario, by means of \emph{channel sounding} (also known as \emph{channel training}). In this technique, the instances of $\bm{A}$ are provided to the receiver by introducing headers in the transmitted matrix $\bm{X} \in R^{n \times \lambda}$, that is, by setting $\bm{X} = \mat{I & {\bm{X}}'}$, where $I \in R^{n \times n}$ is the identity matrix, and ${\bm{X}}' \in R^{n \times (\lambda - n)}$ is a payload matrix coming from a matrix code. For this to work, we clearly need $\lambda_0 \geq n$. Note that channel sounding introduces an overhead of $n^2$ symbols. Nevertheless, the overhead can be made negligible if we are allowed to arbitrarily increase the packet length, that is, the proposed scheme can be capacity-achieving in this asymptotic scenario.

\section{Conclusion} \label{sec:conclusion}

In this work, we investigated coherent multiplicative matrix channels over finite chain rings, which have practical applications in physical-layer network coding. As contributions, we computed the channel capacity, and we determined a necessary and sufficient condition under which a one-shot code can provide shape deficiency correction guarantees. These results naturally generalizes the corresponding ones for finite fields. Furthermore, a coding scheme was proposed, combining several component codes over the residue field to obtain a new composite code over the chain ring. It was shown that if the component codes are suitably chosen, then the composite code is able to achieve the channel capacity and provide shape correction guarantees, both with polynomial time complexity.

Several points are still open. The capacity of the non-coherent MMC, a problem addressed in~\cite{Yang.10.arXiv,Nobrega.13} for the case of finite fields, still needs to be generalized for the case of finite chain rings. Also, designing capacity-achieving coding schemes for the non-coherent MMC with small $\lambda$ is still an open problem, even in the finite-field case.

\appendices \section{Auxiliary Results} \label{sec:auxiliary}

In this appendix, we mention a few basic results that help us compute with $\pi$-adic expansions.

\begin{lemma} \label{lem:pi-adic-aux} Let $x, y, z \in R$. Then, for every $i$, $0 \leq i < s$, we have \begin{enumerate} \item $\left(x \pi^i \right)^{(i+j)} = x^{(j)}$, for $0 \leq j < s-i$; and \item $(x + y \pi^i + z \pi^{i+1})^{(i)} \equiv_\pi x^{(i)} + y^{(0)}$. \end{enumerate} \end{lemma}

\begin{proof} The first claim follows from the uniqueness of the $\pi$-adic expansion. For the second claim, we have \begin{align*} (x + \pi^i y + \pi^{i+1} z)^{(i)} & \quad = \left( \sum_{j=0}^{s-1} \pi^j x^{(j)} + \pi^i \sum_{j=0}^{s-1} \pi^j y^{(j)} + \pi^{i+1} \sum_{j=0}^{s-1} \pi^j z^{(j)}\right)^{(i)} \\ & \quad \overset{(a)} = \left( \sum_{j=0}^{i} \pi^j x^{(j)} + \pi^i y^{(0)} \right)^{(i)} \\ & \quad = \left( \sum_{j=0}^{i-1} \pi^j x^{(j)} + \pi^i (x^{(i)} + y^{(0)}) \right)^{(i)} \\ & \quad \overset{(b)} = \left( \pi^i (x^{(i)} + y^{(0)}) \right)^{(i)} \\ & \quad \overset{(c)} = \left( x^{(i)} + y^{(0)} \right)^{(0)} \equiv_\pi x^{(i)} + y^{(0)}, \end{align*} where $(a)$ follows because factors of $\pi^{i+1}$ do not contribute to the value of the $i$-th term of the $\pi$-adic expansion, $(b)$ is true from the uniqueness of the $\pi$-adic expansion, and $(c)$ follows from the first claim with $j=0$. \end{proof}

\begin{lemma} \label{lem:pi-shape} Let $A \in R^{m \times n}$, and let $\rho = \shape A$. Then, \[ \shape \pi^i A = (\underbrace{0, \ldots, 0}_{i}, \rho_0, \rho_1, \ldots, \rho_{s-i-1}). \] \end{lemma}

\begin{proof} Let $P \in \GL_m(R)$, $Q \in GL_n(R)$, and $D \in R^{m \times n}$ such that $A = PDQ$ and $D$ is the Smith normal form of $A$. Recall that $\shape D = \shape A = \rho$. Then, \[ \shape \pi^i A = \shape \pi^i PDQ = \shape P \pi^iD Q = \shape \pi^iD = (\underbrace{0, \ldots, 0}_{i}, \rho_0, \rho_1, \ldots, \rho_{s-i-1}), \] completing the proof. \end{proof}

\section*{Acknowledgments}

The authors would like to thank Prof.\ Frank R.\ Kschischang for useful discussions.

\end{document}